\documentclass[letterpaper, 10 pt, conference]{ieeeconf}  

\IEEEoverridecommandlockouts 
\usepackage{cite}	     	
\usepackage{textcomp}   	
\usepackage[euler]{textgreek}
\usepackage{lipsum}	     	
\usepackage{soul, color}   	
\usepackage[svgnames]{xcolor} 
\usepackage[utf8]{inputenc} 
\usepackage[T1]{fontenc}    
\let\labelindent\relax  	
\usepackage[shortlabels]{enumitem}  		
\usepackage{microtype}  	
\usepackage[normalem]{ulem}

\usepackage{amsmath}    	
\usepackage{amssymb}		
\usepackage{amsfonts}		

\usepackage{amsthm}    	    
\usepackage{mathtools}		
\usepackage{newtxmath} 	    
\usepackage{bm}             

\usepackage[caption=false,font=footnotesize]{subfig}
\usepackage{graphicx}		
\usepackage{booktabs}   	
\usepackage{siunitx}    
\usepackage[export]{adjustbox}  

\usepackage{hyperref}		

\graphicspath{{figures/}}   
\setlist{noitemsep} 		
\hypersetup{			
    colorlinks,
    linkcolor={DodgerBlue},
    citecolor={DodgerBlue},
    urlcolor={DodgerBlue}
}
\allowdisplaybreaks 

\newcommand{\T}{^{\mathsf{T}}}

\newcommand{\B}[1]{\if#1\relax\bm{#1}\else\mathbf{#1}\fi} 
\newcommand{\R}[1]{\mathrm{#1}}						      
\newcommand{\C}[1]{\mathcal{#1}}	
\newcommand{\BB}[1]{\mathbb{#1}}

\newtheoremstyle{mythmstyle}
{1ex}
{1ex}
{\itshape}
{}
{\bfseries}
{.}
{ }
{\thmname{#1}\thmnumber{ #2}\thmnote{ (#3)}}  

\newtheoremstyle{myremark} 
{}
{}
{}
{}
{\bfseries}
{.}
{.5em}
{}

\theoremstyle{mythmstyle}			
\newtheorem{theorem}{Theorem}

\newtheorem{proposition}[theorem]{Proposition}
\newtheorem{definition}[theorem]{Definition}
\newtheorem{lemma}[theorem]{Lemma}
\theoremstyle{myremark} 
\newtheorem{remark}[theorem]{Remark} 

\title{\LARGE \bf
    Distributed Discontinuous Coupling for Convergence\\in Networks of Heterogeneous Nonlinear Systems
}

\author{Marco Coraggio, Pietro DeLellis and Mario di Bernardo$^{1}$%
\thanks{$^{1}$All the authors are with the Department of Electrical Engineering and Information Technology, University of Naples Federico II, Italy {\tt\small \{marco.coraggio, pietro.delellis, mario.dibernardo\}@unina.it}.}%
\thanks{*The authors wish to acknowledge support from the research project PRIN 2017 ``Advanced Network Control of Future Smart Grids'' funded by the Italian Ministry of University and Research (2020-2023) --- {\tt\small http://vectors.dieti.unina.it}. P. DeLellis also wishes to thanks the University of Naples and Compagnia di San Paolo, Istituto Banco di Napoli, Fondazione for supporting his research under programme ``STAR 2018, project ACROSS.}%
}

\begin{document}

\maketitle
\thispagestyle{empty}
\pagestyle{empty}


\begin{abstract}
Synchronization is a crucial phenomenon in many natural and artificial complex network systems.
Applications include neuronal networks, formation control and coordination in robotics, and frequency synchronization in electrical power grids.
In this paper, we propose the use of a distributed discontinuous coupling protocol to achieve convergence and synchronization in networks of non-identical nonlinear dynamical systems.
We show that the synchronous dynamics is a solution to the average of the nodes' vector fields, and derive analytical estimates of the critical coupling gains required to achieve convergence. Numerical simulations are used to illustrate and validate the theoretical results.
\end{abstract}


\section{Introduction}
\label{sec:section_1}
Coordination, synchronization, formation control and platooning are all examples of emerging phenomena that need to be carefully controlled, maintained, and induced in many applications. 
Examples include frequency synchronization in power grids, formation control and coordination in robotics, cluster synchronization in neuronal networks, and coordination in humans performing joint tasks, e.g \cite{tang2014synchronization, dorfler2014synchronization,oh2015survey}.
In all of these problems, agents are hardly identical, as is often assumed in the literature on complex networks, but are heterogeneous and affected by noise and disturbances. 

The problem of studying the collective behaviour of sets of diffusively coupled non-identical systems was first discussed in \cite{hill2008global} and  later in  \cite{he2013synchronization, delellis2015convergence,montenbruck2015practical, panteley2017synchronization}. 
The emergence of bounded convergence was proven under different conditions showing that, unless the different agents share a common solution (when decoupled) \cite{xiang2007v,zhao2010passivity,zhao2011stability}, or specific symmetries exist in the network structure (see e.g. \cite{zhao2010synchronization}), asymptotic synchronization cannot be achieved, since a unique synchronization manifold does not exist.
Occurrence of partial or cluster synchronization was observed when groups of identical agents can be identified in the ensemble \cite{wang2013cluster}.
Also, a collective behaviour, akin to a ``chimera state'' (where some systems synchronize perfectly, while the others evolve incoherently) \cite{abrams2004chimera}, was investigated in networks of heterogeneous oscillators \cite{laing2009chimera}.
Further results on networks of heterogeneous systems are available in  \cite{seyboth2015robust,grip2012output,chopra2008output} where output- rather than state-synchronization is studied also in the presence of distributed feedback control laws facilitating its emergence.

A crucial open problem is therefore to prove asymptotic convergence in networks of heterogeneous systems with generic structures.
So far, two solutions were proposed that rely on the introduction in the network of some external control actions.
For example, an exogenous input was added onto each node in the network in \cite{lee2013integral,yang2013finite} to achieve this goal, while the use of a self-tuning proportional integral controller was investigated numerically in \cite{burbano2016self}. 

The goal of this paper is to propose an alternative solution to the problem of achieving global asymptotic (rather than bounded) convergence in networks of heterogeneous nonlinear systems.
Differently from previous literature, we prove that, by adding a discontinuous coupling law to the more traditional linear diffusive one, asymptotic convergence can be formally proved, even when the nodes are heterogeneous and do not share a common solution. 
We also show that the synchronous trajectory is a solution to the average of all the individual vector fields of the nodes, and give analytical estimates of the critical values of the coupling gains that guarantee asymptotic synchronization is achieved.
The theoretical derivations are complemented by a set of numerical simulations that show the effectiveness of the proposed approach.
We wish to emphasise that in previous work \cite{cortes2006finite,hui2010finite,liu2015finite} discontinuous communication protocols were used to drive networks of integrators to consensus, but never for networks of generic heterogeneous nonlinear systems.

\section{Problem description and preliminaries}

We consider a generic network of interconnected heterogeneous nonlinear systems of the form 
\begin{equation}\label{eq:network}
\begin{dcases}
\dot{\B{x}}_i(t) = \B{f}_i(\B{x}_i; t) + \B{g}_i(\B{x}; t) \B{u}_i(\B{x}_i; t), \\
\B{y}_i(t) = \B{\phi}_i(\B{x}_i; t),
\end{dcases}
\quad i = 1, \dots, N,
\end{equation}
where $\B{x}_i \in \BB{R}^n$, $\B{u}_i \in \BB{R}^m$, $\B{y}_i \in \BB{R}^l$. 
For the sake of simplicity, we assume that $l = m = n$, $\B{\phi}_i(\B{x}_i; t) = \B{x}_i$ and $\B{g}_i(\B{x}_i; t) = \B{I}_n$, with $\B{I}_n$ being the $n$-dimensional identity matrix.

\textit{Control objective.} \quad 
We seek a distributed coupling protocol $\B{u}_i$ that, under suitable assumptions on the vector fields of the agents and on the network structure, drives all nodes towards \emph{global asymptotic synchronization}, that is, it guarantees that, for all initial conditions $\B{x}_i(t=0) \in\BB{R}^n$, $i=1, \ldots, N$, 
\[
\lim_{t \rightarrow +\infty} \left\lVert \B{x}_i(t) - \B{x}_j(t) \right\rVert = 0,  
\quad i, j = 1, \dots, N,
\]
where $\left\lVert \cdot \right\rVert_p$ is the $p$-norm operator, with $p=2$ if it is omitted.

\textit{Control design.} \quad 
To achieve the control objective stated above, we will show that, under certain conditions, asymptotic convergence is guaranteed by the following distributed coupling law:
\begin{equation}\label{eq:diffusive_discontinuous_coupling}
\B{u}_i = - c \sum_{j=1}^N L_{ij} \B{\Gamma} \left( \B{x}_j - \B{x}_i \right)
- c_\R{d} \sum_{j=1}^N L_{ij}^\R{d} \B{\Gamma}_\R{d} \R{sign}\left( \B{x}_j - \B{x}_i \right),
\end{equation}
where $L_{ij}, L_{ij}^\R{d}$ are the $(i,j)$-th elements of the Laplacian matrices $\B{L}, \B{L}_\R{d}$ describing two undirected unweighted graphs, $\C{G} = (\C{V}, \C{E})$ and $\C{G}_\R{d} = (\C{V}, \C{E}_\R{d})$; $\C{V}$ being the set of vertices, and $\C{E}$, $\C{E}_\R{d}$ the sets of edges.
The matrices $\B{\Gamma}, \B{\Gamma}^\R{d} \in \BB{R}^{n \times n}$, also known as \emph{inner coupling matrices}, are assumed to be positive semi-definite.
Finally, the sign of a vector is to be intended as $\R{sign}(\B{v}) = [ \R{sign}(v_1) \ \cdots \ \R{sign}(v_n) ]\T \in \BB{Z}^n$, for $\B{v} \in \BB{R}^n$.

\textit{Preliminary definitions and lemmas.} \quad
We define the state average $\tilde{\B{x}} \triangleq \frac{1}{N} \sum_{i=1}^N \B{x}_i$ and the synchronization errors $\B{e}_i \triangleq \B{x}_i - \tilde{\B{x}}$, for $i = 1, \dots, N$, and
introduce the stack vectors 
$\bar{\B{x}} \triangleq [ \B{x}_1\T \ \cdots \ \B{x}_N\T ]\T$, 
$\bar{\B{u}} \triangleq [ \B{u}_1\T \ \cdots \ \B{u}_N\T ]\T$, and 
$\bar{\B{y}} \triangleq [ \B{y}_1\T \ \cdots \ \B{y}_N\T ]\T$.
We denote a closed ball about some point $\B{v}$ of radius $r$ as $\C{B}_r^\R{c}(\B{v})$, dropping the argument when $\B{v}$ is the origin.

\begin{definition}[\hspace{-0.01cm}\cite{coraggio2019achieving}]\label{def:mu_infinity_minus}
    Given a matrix $\B{A} \in \BB{R}^{n \times n}$, we define the quantity $\mu_\infty^-(\B{A})$ as
    \begin{equation}
    \mu_\infty^-(\B{A}) \triangleq \min_{i = 1, \dots, n} \Bigg( A_{ii} - \sum_{j = 1, j \ne i}^n \left\lvert A_{ij} \right\rvert \Bigg).
    \label{eq:mu_infty-}
    \end{equation}
\end{definition}

\begin{definition}[QUADness \cite{delellis2011quad}]\label{def:QUAD}
    A vector field $\B{f} : \BB{R}^n \times \BB{R}_{\ge 0} \rightarrow \BB{R}^n$ is said to be \emph{QUAD($\B{P}$, $\B{Q}$)} if there exist matrices $\B{P}, \B{Q} \in \BB{R}^{n \times n}$ such that, for all $\B{v}_1, \B{v}_2 \in \BB{R}^n$, $t \in \BB{R}_{\ge 0}$,
    \begin{equation*}
    \left( \B{v}_1 - \B{v}_2 \right)\T \B{P} \left[ \B{f}(\B{v}_1; t) - \B{f}(\B{v}_2; t) \right] \le
    \left( \B{v}_1 - \B{v}_2 \right)\T \B{Q} \left( \B{v}_1 - \B{v}_2 \right).
    \end{equation*}
\end{definition}
\begin{lemma}[\hspace{-0.01cm}\cite{gallegos2015fractional}]\label{lem:modified_barbalat}
    Let $f$ be a scalar non-negative uniformly continuous function of time, and let $C > 0$. 
    If, for all $t \ge 0$, $\int_0^t f(\tau) \ \R{d}\tau < C$, then $\lim_{t \rightarrow +\infty} f(t) = 0$.
\end{lemma}


\begin{definition}[Uniform asymptotic boundedness]\label{def:asymptotic_boundedness} 
A nonlinear system of the form \eqref{eq:network} with a given input function $\B{u}_i(\B{x}_i;t)$ is \emph{uniformly asymptotically bounded to $\C{B}_r^\R{c}$} if there exists $r \in \BB{R}_{>0}$ such that, for all initial conditions,
    \begin{equation}
    \mathop{\lim \sup}_{t \rightarrow + \infty} \left\lVert \B{x}_i(t) \right\rVert \le r.
    \end{equation}
\end{definition}

\begin{definition}[Uniform ultimate boundedness]\label{def:ultimate_boundedness}
A nonlinear system of the form \eqref{eq:network} with a given input function $\B{u}_i(\B{x}_i;t)$ is \emph{uniformly ultimately bounded to $\C{B}_r^\R{c}$}, with $r \in \BB{R}_{>0}$, if there exists a function  $T : \BB{R}^{n} \rightarrow [0, +\infty[ {}$ such that 
    \begin{equation}
    \forall t \ge T(\B{x}_i(0)), \quad
    \left\lVert \B{x}_i(t) \right\rVert \le r.
    \end{equation}
\end{definition}
It is important to remark that if a dynamical system is uniformly asymptotically bounded to $\C{B}_r^\R{c}$, then it is also uniformly ultimately bounded to $\C{B}_{r^+}^\R{c}$, for any $r^+ > r$.





Next, we extend the concept of semipassivity \cite{pogromsky1999diffusion} to nonlinear systems in the presence of a discontinuous input by adapting the definition of passivity for non-smooth systems in \cite{nakakuki2006remark}.%
\footnote{In Definition \ref{def:semipassivity_discontinuous}, to ensure the existence of a solution, we assume the Filippov vector field defining the system is locally bounded, takes nonempty, compact, and convex values and is upper-semicontinuous; \cite[Proposition S2]{cortes2008discontinuous}.%
}
\begin{definition}[Semipassivity with a discontinuous input]\label{def:semipassivity_discontinuous}
    A nonlinear system of the form (\ref{eq:network}) subject to a discontinuous input $\B{u}_i(\B{x}_i,t)$ in $\B{x}_i$ is \emph{semipassive} if the following conditions hold:
    \begin{enumerate}[(a)]
        \item there exist $\rho_i > 0$, a continuous  function $\alpha_i : {} [\rho_i, +\infty[ {} \rightarrow \BB{R}_{\ge 0}$, and a continuous function $h_i : \BB{R}^n \rightarrow \BB{R}$, termed as the \emph{stability component}, such that
      \begin{equation}\label{eq:stability_component}
        h_i(\B{x}_i) \ge \alpha_i(\left\lVert \B{x}_i \right\rVert) \ge 0,
        \quad \text{if } \left\lVert \B{x}_i \right\rVert \ge \rho_i;
        \end{equation}    
        \item there exists a continuous non-negative \emph{storage function} $V_i : \BB{R}^n \rightarrow \BB{R}_{\ge 0}$ such that $V_i(\B{0}) = 0$ and
        \begin{equation}
        V_i(\B{x}_i(t)) - V_i(\B{x}_i(t_0)) \le p_i(t; \B{x}_i(t_0)),
        \end{equation}
        where $p_i(t; \B{x}_i(t_0))$ is the Filippov solution at time $t$, starting from initial condition $p_i(t_0; \B{x}_i(t_0)) = 0$, given $\B{x}_i(t_0)$, to the differential equation
        \begin{equation*}
        \dot{p}_i(t; \B{x}_i(t)) = \left( \B{y}_i (\B{x}_i(t)) \right)\T \B{u}_i(\B{x}_i(t)) - h_i(\B{x}_i(t)).
        \end{equation*}
    \end{enumerate}
    Moreover, if the function $\alpha_i$ is strictly positive for $\left\lVert \B{x}_i \right\rVert > \rho_i$, then \eqref{eq:network} is said to be \emph{strictly semipassive}.
    Also, if $\alpha_i$ is radially unbounded and increasing, then  \eqref{eq:network} is said to be \emph{strongly strictly semipassive}.
\end{definition}


\section{Boundedness of heterogeneous networks}

In this Section, we prove uniform asymptotic boundedness by exploiting Lemma \ref{lem:fix_pogromsky_discontinuous} (see Appendix) and following the steps in \cite{pogromsky1999diffusion}.
Then, in Section \ref{sec:asymptotic_convergence}, we move to proving asymptotic convergence. 

%
%
%


\begin{proposition}\label{prp:pogromsky_for_disc}
Consider network \eqref{eq:network}-\eqref{eq:diffusive_discontinuous_coupling}.
If 
\begin{enumerate}[(a)]
    \item all systems in \eqref{eq:network} are strongly strictly semipassive, with stability components $h_i$, $i = 1, \dots, N$;
    \item all systems in \eqref{eq:network} have radially unbounded storage functions $V_i$;
    \item $c \ge 0$, $c_\R{d} \ge 0$, $\R{sym}( \B{\Gamma} ) \ge 0$, and $\mu_\infty^-(\B{\Gamma}_{\R{d}}) \ge 0$;
\end{enumerate} 
then \eqref{eq:network}-\eqref{eq:diffusive_discontinuous_coupling} is uniformly asymptotically bounded.
\end{proposition}
\begin{proof}
    Consider the function $\bar{V} : \BB{R}^{Nn} \rightarrow \BB{R}_{\ge 0}$ given by
    \begin{equation}\label{eq:V_bar}
    \bar{V}(\bar{\B{x}}) \triangleq V_{1}\left(\B{x}_{1}\right)+ \ldots+ V_{N}\left(\B{x}_{N}\right).
    \end{equation}
    Since $\bar{V}$ is the sum of radially unbounded functions, it is radially unbounded itself.
    From \eqref{eq:V_bar} and Definition \ref{def:semipassivity_discontinuous}, we have
    \begin{equation}\label{eq:network_semipassivity}
    \bar{V}(\bar{\B{x}}(t)) - \bar{V}(\bar{\B{x}}(0)) \le \bar{p}(t; \bar{\B{x}}(0)),
    \end{equation}
    where $\bar{p}(t; \bar{\B{x}}(t_0)) \triangleq \sum_{i=1}^N p_i(t; \B{x}_i(t_0))$. 
    
    \begin{figure}[t]
        \centering
        \subfloat[]{\includegraphics[scale=0.9]{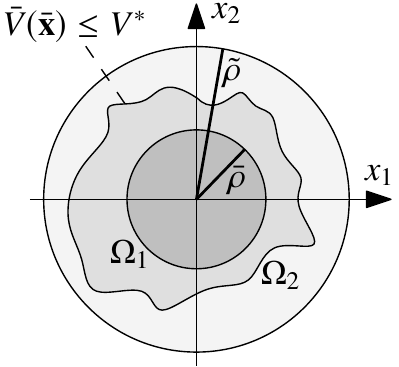}
        \label{fig:sets_V}} \hfill
        \subfloat[]{\includegraphics[scale=0.9]{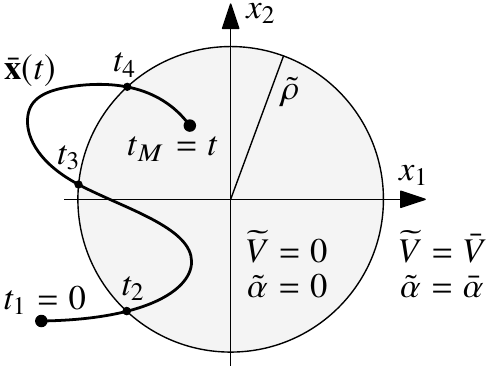}
        \label{fig:time_instants}}%
        \caption{Example of sets (a) and time instants (b) described in the proof of Proposition \ref{prp:pogromsky_for_disc} with $n = 1$, $N = 2$.}
        \label{fig:figures_proof}
    \end{figure}
    
    Note that, given the hypotheses of this Proposition, Lemma \ref{lem:fix_pogromsky_discontinuous} (see Appendix) holds.
    Then, consider the set $\Omega_1 \triangleq \{ \bar{\B{x}} \mid \left\lVert \bar{\B{x}} \right\rVert \le \bar{\rho} \}$, which is compact and where $\bar{\rho}$ is given by the Lemma. 
    Since $\bar{V}$ is continuous and radially unbounded, we can find a scalar $V^* > 0$ such that the compact set $\Omega_2 \triangleq \{ \bar{\B{x}} \mid \bar{V}(\bar{\B{x}}) \leq V^*\}$ fulfils $\Omega_2 \supset \Omega_1$. 
    As $\Omega_2$ is compact, there exists a closed ball of the origin with radius $\tilde{\rho} \ge \bar{\rho}$ that contains $\Omega_2$; see the sketch diagram reported in Fig.~\ref{fig:sets_V} for the case that $n=1$, $N=2$.
    Now, we define the functions
    \begin{equation}\label{eq:V_prime}
    \widetilde{V}(\bar{\B{x}}) \triangleq 
    \begin{dcases}
    0, & \text {if } \left\lVert \bar{\B{x}} \right\rVert \leq \tilde{\rho}, \\
    \bar{V}(\bar{\B{x}}), & \text{otherwise},
    \end{dcases}
    \end{equation}
    \begin{equation}\label{eq:proof_step_14}
    \tilde{\alpha}(\left\lVert \bar{\B{x}} \right\rVert) \triangleq 
    \begin{dcases}
    0, & \text {if } \left\lVert \bar{\B{x}} \right\rVert \leq \tilde{\rho}, \\
    \bar{\alpha}(\left\lVert \bar{\B{x}} \right\rVert), & \text{otherwise}.
    \end{dcases}
    \end{equation}
    Next, we divide the generic time interval $[0, t]$ in $M-1$ contiguous sub-intervals $[t_1=0, t_2], \dots, [t_{M-1}, t_M=t]$, where $t_2 \dots, t_{M-1}$ are the time instants at which $\bar{\B{x}}$ crosses transversely the level set where $\left\lVert \bar{\B{x}} \right\rVert = \tilde{\rho}$ (see Fig. \ref{fig:time_instants}). 
    With this partition of the time interval $[0, t]$ we have that, in each sub-interval $[t_{j-1}, t_j]$, either
    \begin{equation}\label{eq:proof_step_15}
    \widetilde{V}(\bar{\B{x}}(t_j)) - \widetilde{V}(\bar{\B{x}}(t_{j-1})) = 0,
    \end{equation}
    because of \eqref{eq:V_prime}, or
    \begin{equation}\label{eq:proof_step_12}
    \widetilde{V}(\bar{\B{x}}(t_j)) - \widetilde{V}(\bar{\B{x}}(t_{j-1})) \le \bar{p}(t_j; \bar{\B{x}}(t_{j-1})),
    \end{equation}
    because of \eqref{eq:network_semipassivity}.
    Now, note that 
    $\dot{\bar{p}} (\bar{\B{x}}) = - \bar{q}(\bar{\B{x}})$; $\bar{q}$ being defined in Lemma \ref{lem:fix_pogromsky_discontinuous}.
    By exploiting the Lemma, with $\bar{\alpha}$ defined therein, we have%
    \begin{equation}\label{eq:proof_step_11}
    \dot{\bar{p}}(\bar{\B{x}}) = - \bar{q}(\bar{\B{x}}) \le - \bar{\alpha}(\left\lVert \bar{\B{x}} \right\rVert).
    \end{equation}
    From \eqref{eq:proof_step_14} and \eqref{eq:proof_step_11}, it follows that
    \begin{equation}\label{eq:proof_step_13}
    \bar{p}(t_j; \bar{\B{x}}(t_{j-1})) \le 
    - \int_{t_{j-1}}^{t_j} \bar{\alpha}(\left\lVert \bar{\B{x}} (\tau)\right\rVert) \R{d}\tau =
    - \int_{t_{j-1}}^{t_j} \tilde{\alpha}(\left\lVert \bar{\B{x}} (\tau)\right\rVert) \R{d}\tau.
    \end{equation}
    Combining \eqref{eq:proof_step_12} and \eqref{eq:proof_step_13}, and from Lemma \ref{lem:fix_pogromsky_discontinuous}, we have
    \begin{multline}\label{eq:bound_single_interval}
    \widetilde{V}(\bar{\B{x}}(t_j)) - \widetilde{V}(\bar{\B{x}}(t_{j-1})) \le - \bar{p}(t_j, \bar{\B{x}}(t_{j-1})) \le \\
    - \int_{t_{j-1}}^{t_j} \tilde{\alpha}(\left\lVert \bar{\B{x}} (\tau)\right\rVert) \ \R{d}\tau \le 0.
    \end{multline}
    Therefore, since
    \begin{multline}\label{eq:V_time_intervals}
    \widetilde{V}(\bar{\B{x}}(t)) - \widetilde{V}(\bar{\B{x}}(0)) =
    [ \widetilde{V}(\bar{\B{x}}(t)) - \widetilde{V}(\bar{\B{x}}(t_{M-1})) ] + [ \widetilde{V}(\bar{\B{x}}(t_{M-1})) \\ - \widetilde{V}(\bar{\B{x}}(t_{M-2})) ] + \ldots + [ \widetilde{V}(\bar{\B{x}}(t_{2})) - \widetilde{V}(\bar{\B{x}}(0)) ],
    \end{multline}   
    exploiting \eqref{eq:proof_step_14}, \eqref{eq:proof_step_15} and \eqref{eq:bound_single_interval}, we get
    \begin{equation}\label{eq:V_prime_q}
    \widetilde{V}(\bar{\B{x}}(t)) - \widetilde{V}(\bar{\B{x}}(0)) \le - \int_{0}^{t} \tilde{\alpha}(\left\lVert \bar{\B{x}} (\tau)\right\rVert) \ \R{d}\tau \le 0.
    \end{equation}
    Hence, $\widetilde{V}(\bar{\B{x}}(t)) \le \widetilde{V}(\bar{\B{x}}(0))$, i.e. $\widetilde{V}(\bar{\B{x}}(t))$ is bounded for all $t \ge 0$.    
    Also, for large values of $\bar{\B{x}}$ ($\left\lVert \bar{\B{x}} \right\rVert > \tilde{\rho}$), from \eqref{eq:V_prime} we have $\widetilde{V}(\bar{\B{x}}) = \bar{V}(\bar{\B{x}})$; therefore $\widetilde{V}(\bar{\B{x}})$ is radially unbounded as $\bar{V}(\bar{\B{x}})$ is.
    Thus, $\widetilde{V}(\bar{\B{x}}(t))$ being bounded implies that $\bar{\B{x}}$ must be bounded (even if $\widetilde{V}$ is a discontinuous function).    
    This means that network  \eqref{eq:network}-\eqref{eq:diffusive_discontinuous_coupling} is \emph{Lagrange stable}, i.e. $\left\lVert \B{x}(t) \right\rVert < +\infty$ for all $t$.
    
    Next, we show that \eqref{eq:network}-\eqref{eq:diffusive_discontinuous_coupling} is uniformly asymptotically bounded.
    We define
    \begin{equation}
    \tilde{\alpha}'(\left\lVert \bar{\B{x}} \right\rVert) \triangleq 
    \begin{dcases}
    0, & \text {if } \left\lVert \bar{\B{x}} \right\rVert \leq \tilde{\rho}, \\
    \bar{\alpha}(\left\lVert \bar{\B{x}} \right\rVert) - \bar{\alpha}(\tilde{\rho}), & \text{otherwise},
    \end{dcases}
    \end{equation}
    which is continuous and null if and only if $\left\lVert \bar{\B{x}} \right\rVert \leq \tilde{\rho}$, as $\bar{\alpha}$ is increasing.
    In addition, since the network solutions are bounded, $\bar{\B{x}}(t)$ belongs to a compact set, and therefore $\tilde{\alpha}'(\left\lVert \bar{\B{x}}(t) \right\rVert)$ is uniformly continuous in that set.
    From \eqref{eq:V_prime_q}, we know that $\int_{0}^{t} \tilde{\alpha}(\left\lVert \bar{\B{x}}(\tau) \right\rVert) \R{d}\tau$ is finite for all $t \in [0, +\infty]$ as it is bounded by two finite terms.
    Consequently, 
    $\int_{0}^{t} \tilde{\alpha}'(\left\lVert \bar{\B{x}}(\tau) \right\rVert) \R{d}\tau$ 
    is also bounded, and we can employ Lemma \ref{lem:modified_barbalat} to conclude that $\lim_{t \rightarrow + \infty} \tilde{\alpha}'(\left\lVert \bar{\B{x}}(t) \right\rVert) = 0$.
    Since $\tilde{\alpha}'(\left\lVert \bar{\B{x}} \right\rVert)$ is null only when $\left\lVert \bar{\B{x}} \right\rVert \leq \tilde{\rho}$, this means that 
    \begin{equation}
    \mathop{\lim \sup}_{t \rightarrow + \infty} \left\lVert \bar{\B{x}}(t) \right\rVert \le \tilde{\rho}.
    \end{equation}
\end{proof}

\section{Asymptotic convergence of heterogeneous networks}
\label{sec:asymptotic_convergence}
Before giving our main result,
we define the average vector field $\tilde{\B{f}} : \BB{R}^{nN} \rightarrow \BB{R}^n$ as follows:
\begin{equation}\label{eq:average_dynamics}
\tilde{\B{f}}(\bar{\B{x}}) \triangleq \frac{1}{N} \sum_{i = 1}^{N} \B{f}_i({\B{x}_i}) = \dot{\tilde{\B{x}}},
\end{equation}
where the coupling terms in $\dot{\tilde{\B{x}}}$ cancel out since $\B{L},\B{L}_\R{d}$ are symmetric.
Recalling that $\B{e}_i \triangleq \B{x}_i - \tilde{\B{x}}$, we can write 
\begin{equation}\label{eq:error_dynamics}
\begin{aligned}
\dot{\B{e}}_i = \dot{\B{x}}_i - \dot{\tilde{\B{x}}}&=
\B{f}_i(\B{x}_i)
- c \sum_{j=1}^{N} L_{ij} \B{\Gamma} (\B{x}_j - \B{x}_i)\\
 & - c_\R{d} \sum_{j=1}^{N} L_{ij}^\R{d} \B{\Gamma}_\R{d} \R{sign} (\B{x}_j - \B{x}_i) - \tilde{\B{f}}(\bar{\B{x}}).
\end{aligned}
\end{equation}

\begin{theorem}\label{thm:stability_heterogeneous_networks}
    Consider network \eqref{eq:network} controlled by the distributed control action \eqref{eq:diffusive_discontinuous_coupling}.
    If
    \begin{enumerate}[(a)]
        \item the controlled network is uniformly ultimately bounded to the ball $\C{B}_r^\R{c}$, for some $r > 0$;
        \item each agent dynamics $\B{f}_i$ is QUAD($\B{P}$, $\B{Q}_i$) in $\C{B}_r^\R{c}$, and $\R{sym}(\B{P} \B{\Gamma}) > 0$, $\mu_\infty^-(\B{P} \B{\Gamma}_{\R{d}}) > 0$;
        \item $\C{G}$ and $\C{G}_\R{d}$ are connected graphs;
    \end{enumerate}
    then 
    \begin{enumerate}[(i)]
        \item there exist  $c^*$ and $c_\R{d}^*$ such that, if $c > c^*$ and $c_\R{d} \ge c_\R{d}^*$, then global asymptotic synchronization is achieved.
        Moreover, the asymptotic synchronous trajectory $\B{s}(t)$ is a solution to $\dot{\B{s}}(t) = \frac{1}{N} \sum_{i=1}^N \B{f}_i(\B{s}(t))$;
        \item $c^*$ and $c_\R{d}^*$ are given by 
        \begin{equation}\label{eq:c_star}
        c^* \triangleq \frac{\max_i{\left(\lVert \B{Q}_i \rVert_2\right)}} {\lambda_{2} (\B{L}) \lambda_{\R{min}} (\R{sym} \ \B{P}\B{\Gamma} )},
        \quad 
        c^*_\R{d} \triangleq \frac{\left\lVert \left( \left\lvert \B{P} \right\rvert \right) \B{m} \right\rVert_\infty}{\delta_{\C{G}_\R{d}} \mu_\infty^-(\B{P}\B{\Gamma}_\R{d})},
        \end{equation}
        where $\delta_{\C{G}_\R{d}}$ is the \emph{minimum density} \cite{coraggio2019achieving} of the graph $\C{G}_\R{d}$, and $\B{m} \in \BB{R}_{\ge 0}^n$ is a vector such that
        \begin{equation}
        \B{m} \ge \left\lvert \B{f}_i(\tilde{\B{x}}) - \tilde{\B{f}}(\bar{\B{x}}) \right\rvert,
        \quad \forall i \in \{1, \dots, N\}, \ \forall \bar{\B{x}} \in \C{B}_r^\R{c}.
        \end{equation}
    \end{enumerate}
\end{theorem}

\begin{proof}
    Consider the candidate common Lyapunov function 
    $
    V \triangleq \frac{1}{2} \sum_{i=1}^N \B{e}_i\T \B{P} \B{e}_i.
    $
    From \eqref{eq:error_dynamics}, we have
    \begin{equation}
    \begin{aligned}
    \dot{V} &= \sum_{i=1}^N \B{e}_i\T \B{P} \left( \B{f}_i(\B{x}_i) - \tilde{\B{f}}(\bar{\B{x}}) \right)
    - c \sum_{i=1}^N \sum_{j=1}^N L_{ij} \B{e}_i\T \B{P} \B{\Gamma} \B{e}_j \\
    &\phantom{=\,}- c_\R{d} \sum_{i=1}^N \sum_{j = 1}^{N} L_{ij}^\R{d} \B{e}_i\T \B{P} \B{\Gamma}_\R{d} \R{sign}(\B{e}_j - \B{e}_i),
    \end{aligned}
    \end{equation}
    where we used the fact that $\R{sign}(\B{x}_j - \B{x}_i) = \R{sign}(\B{e}_j - \B{e}_i)$.
    Then, adding and subtracting $\sum_{i=1}^{N} \B{e}_i\T \B{P} \B{f}_i(\tilde{\B{x}})$, we have
    \begin{multline*}
    \dot{V} = \sum_{i=1}^N \B{e}_i\T \B{P} \left( \B{f}_i(\B{x}_i) - \B{f}_i(\tilde{\B{x}}) \right) 
    +\sum_{i=1}^N \B{e}_i\T \B{P} \left( \B{f}_i(\tilde{\B{x}}) - \tilde{\B{f}}(\bar{\B{x}}) \right) \\
    -c \sum_{i=1}^N \sum_{j=1}^N L_{ij} \B{e}_i\T \B{P} \B{\Gamma} \B{e}_j
    - c_\R{d} \sum_{i=1}^N \sum_{j = 1}^{N} L_{ij}^\R{d} \B{e}_i\T \B{P} \B{\Gamma}_\R{d} \R{sign}(\B{e}_j - \B{e}_i).
    \end{multline*}
    In addition, since the communication graphs are undirected ($L_{ij}^\R{d} = L_{ji}^\R{d}$), for each term $\B{e}_i\T \B{P} \B{\Gamma}_\R{d} \R{sign}(\B{e}_j - \B{e}_i)$, there must exist the symmetric term $\B{e}_j\T \B{P} \B{\Gamma}_\R{d} \R{sign}(\B{e}_i - \B{e}_j)$.
    Hence, we may recast $\dot{V}$ as
    \begin{multline*}
    \dot{V} = \sum_{i=1}^N \B{e}_i\T \B{P} \left( \B{f}_i(\B{x}_i) - \B{f}_i(\tilde{\B{x}})\right)
    +\sum_{i=1}^N \B{e}_i\T \B{P} \left( \B{f}_i(\tilde{\B{x}}) - \tilde{\B{f}}(\bar{\B{x}}) \right) \\
    -c \sum_{i=1}^N \sum_{j=1}^N L_{ij} \B{e}_i\T \B{P} \B{\Gamma} \B{e}_j
    - c_\R{d} \sum_{(i,j) \in \C{E}_\R{d}}  (\B{e}_i - \B{e}_j)\T \B{P} \B{\Gamma}_\R{d} \R{sign}(\B{e}_i - \B{e}_j).
    \end{multline*}
    
    As the network is uniformly ultimately bounded, there exists a finite $T^* > 0$ such that, for $t \ge T^*$, $\left\lVert \B{x}(t) \right\rVert \in \C{B}_r^\R{c}$.
    From now on, we take $t \ge T^*$, and,
    since $\B{f}_i$ is QUAD($\B{P}$, $\B{Q}_i$), we get
    \begin{multline}
    \dot{V} \le \sum_{i=1}^N \left( \B{e}_i\T \B{Q}_i \B{e}_i \right)
    +\sum_{i=1}^N \B{e}_i\T \B{P} \left( \B{f}_i(\tilde{\B{x}}) - \tilde{\B{f}}(\bar{\B{x}}) \right)
    - c \sum_{i=1}^N \sum_{j=1}^N L_{ij} \\
    \B{e}_i\T \B{P} \B{\Gamma} \B{e}_j
    - c_\R{d} \sum_{(i,j) \in \C{E}_\R{d}}  (\B{e}_i - \B{e}_j)\T \B{P} \B{\Gamma}_\R{d} \R{sign}(\B{e}_i - \B{e}_j).
    \end{multline}
    By defining the diagonal block matrix $\bar{\B{Q}}$ having $\B{Q}_1, \dots, \B{Q}_N$ on its diagonal, we can write $\sum_{i=1}^N \left( \B{e}_i\T \B{Q}_i \B{e}_i \right) = \bar{\B{e}}\T \bar{\B{Q}} \bar{\B{e}}$.
    
    As all $\B{f}_i$'s are QUAD in $\C{B}_r^\R{c}$, they are also bounded therein.
    Then, there exists a vector $\B{m} \in \BB{R}^n_{\ge 0}$, such that
    \begin{equation}
    \B{m} \ge \left\lvert \B{f}_i(\tilde{\B{x}}) - \tilde{\B{f}}(\bar{\B{x}}) \right\rvert,
    \qquad \forall i \in \{1, \dots, N\}, \ \forall \bar{\B{x}} \in \C{B}_r^\R{c}.
    \end{equation}
    Therefore, letting $M \triangleq \left\lVert \left( \left\lvert \B{P} \right\rvert \right) \B{m} \right\rVert_\infty$, it holds that
    \begin{equation}
    \begin{aligned}
    \sum_{i=1}^N \B{e}_i\T \B{P} \left( \B{f}_i(\tilde{\B{x}}) - \tilde{\B{f}}(\bar{\B{x}}) \right) &\le
    \sum_{i=1}^N \left\lVert \B{e}_i \right\rVert_1 \left\lVert \B{P} \left( \B{f}_i(\tilde{\B{x}}) - \tilde{\B{f}}(\bar{\B{x}}) \right) \right\rVert_\infty \\ &\le
    M \sum_{i=1}^N \left\lVert \B{e}_i \right\rVert_1  =
    M \left\lVert \bar{\B{e}} \right\rVert_1.
    \end{aligned}
    \end{equation}
    Defining $\bar{\B{a}} \triangleq \left(\B{B}_\R{d}\T \otimes \B{I}_n \right) \bar{\B{e}}$, we obtain $\dot{V} \le W_1 + W_2$, where 
    \begin{align}
        W_1 &\triangleq \bar{\B{e}}\T \left( \bar{\B{Q}} - c \B{L} \otimes \B{P} \B{\Gamma} \right) \bar{\B{e}},\\
        W_2 &\triangleq M \left\lVert \bar{\B{e}} \right\rVert_1 
        - c_\R{d} \bar{\B{a}}\T \left( \B{I}_{N_{\C{E}_\R{d}}} \otimes \B{P} \B{\Gamma}_\R{d} \right) \R{sign} (\bar{\B{a}}).
    \end{align}
    Then, following the steps in \cite[proof of Theorem 5]{coraggio2019achieving}, we find that $W_1 < 0$ if $c > c^*$, and $W_2 \le 0$ if $c_\R{d} \ge c_\R{d}^*$, with $c^*, c_\R{d}^*$ given by \eqref{eq:c_star}.
    Finally, since $W_1 < 0$ and $W_2 \le 0$, then $\dot{V} < 0$, which means that all $\B{e}_i$'s tend to zero, i.e. all $\B{x}_i$'s tend to $\tilde{\B{x}}$, whose dynamics is given in \eqref{eq:average_dynamics}. 
\end{proof}
\begin{remark}
Note that the assumptions on boundedness and QUADness in Theorem \ref{thm:stability_heterogeneous_networks} are quite mild and they can be easily verified.
Indeed, uniform ultimate boundedness of network \eqref{eq:network} can be checked by using Proposition \ref{prp:pogromsky_for_disc}, while the QUADness hypothesis on the dynamics can be verified by testing boundedness of the Jacobian of the individual vector fields; see Proposition \ref{pro:bounded_jacobian_quad}.
\end{remark}
\begin{remark}
Theorem \ref{thm:stability_heterogeneous_networks} can be easily adapted to account for possible discontinuities in the nodes' dynamics.
In that case, the agents must be \textsigma-QUAD($\B{P}$, $\B{Q}_i$, $\B{M}_i$) \cite{coraggio2019achieving} (rather than QUAD) and the critical threshold for the discontinuous coupling layer can be proved to be
\begin{equation}
c^*_\R{d} \triangleq \frac{\left\lVert \left( \left\lvert \B{P} \right\rvert \right) \B{m} \right\rVert_\infty + \left\lVert \bar{\B{M}} \right\rVert_\infty} {\delta_{\C{G}_\R{d}} \mu_\infty^-(\B{P}\B{\Gamma}_\R{d})},
\quad 
\bar{\B{M}} \triangleq \begin{bmatrix}
\B{M}_1 &        & \\
& \ddots & \\
&		 & \B{M}_N
\end{bmatrix}.
\end{equation}
\end{remark}

\section{Numerical validation}

We consider a set of 3 modified van der Pol oscillators of the form
\begin{equation}\label{eq:modified_van_der_pol}
\dot{\B{x}} = \B{f}_i(\B{x})  + \B{u} = 
\begin{bmatrix}
x_1 - \epsilon x_1 \\
\mu_i (1 - x_1^2 - \eta x_2^2) x_2 - x_1 \\
\end{bmatrix}
+ \begin{bmatrix} u_1 \\ u_2
\end{bmatrix},
\end{equation}
for $i = 1,2,3$, with $\epsilon = 0.01$, $\eta = 0.001$, and $\mu_1 = 1$, $\mu_2 = 2$, $\mu_3 = 3$.
We couple the agents through the diffusive and discontinuous coupling law \eqref{eq:diffusive_discontinuous_coupling}, with $\B{L}$, $\B{L}_\R{d}$ corresponding to complete graphs, and $\B{\Gamma} = \B{\Gamma}_\R{d} = \B{I}_2$.
Introducing the storage function $V_i(\B{x}) = \frac{1}{2}(x_1^2 + x_2^2)$, we can show systems \eqref{eq:modified_van_der_pol} are strongly strictly semipassive. 
Indeed,
\begin{equation*}
\begin{split}
\dot{V}_i &= x_1 \dot{x}_1 + x_2 \dot{x}_2 \\
&= x_1 x_2 - \epsilon x_1^2 + x_1 u_1 + \mu_i x_2^2 (1 - x_1^2 - \eta x_2^2 ) - x_1 x_2 + x_2 u_2 \\
&= - \epsilon x_1^2 + \mu_i x_2^2 (1 - x_1^2 - \eta x_2^2 ) + \B{x}\T \B{u} = - h_i(\B{x}) + \B{y}\T \B{u},
\end{split}
\end{equation*}
where $h_i(\B{x}) \triangleq \epsilon x_1^2 + \mu_i x_2^2 ( x_1^2 + \eta x_2^2 - 1 )$.
From Proposition \ref{prp:pogromsky_for_disc}, it follows that the network is uniformly ultimately bounded to $\C{B}_r^\R{c}$ for some $r$; a numerical exploration shows that $r = 7.72$ is a suitable value.
Since $\B{f}$ is continuous, its Jacobian is bounded in $\C{B}_r^\R{c}$, and the three agents are QUAD($\B{I}$, $\B{Q}_i$), $i = 1, 2, 3$ (see Proposition \ref{pro:bounded_jacobian_quad} in the Appendix), 
All the assumptions of Theorem \ref{thm:stability_heterogeneous_networks} are fulfilled, and its thesis can be used to compute the critical values $c^*$ and $c_\R{d}^*$ that guarantee asymptotic synchronization.
Specifically, knowing $r$, we can compute analytically that  $\max_i{\left(\lVert \B{Q}_i \rVert_2\right)} \approx 11.58$, and numerically that $\left\lVert \B{m} \right\rVert_\infty \approx 179.90$; moreover, $\lambda_2(\B{L}) = N = 3$, and $\delta_{\C{G}_\R{d}} = N/2 = 3/2$ \cite{coraggio2019achieving}. 
Therefore, through \eqref{eq:c_star}, we compute that $c^* = 3.86$ and $c_\R{d}^* = 119.93$. 

In Fig.~\ref{fig:oscillators}, two simulations are reported.
Namely, in Fig.~\ref{fig:oscillators_1}, where $c = 4 > c^*$ and the discontinuous coupling is absent, the network does not achieve synchronization.
When the discontinuous action is turned on with strength $c_\R{d} = 120 > c_{\R{d}}^*$ in Fig.~\ref{fig:oscillators_2}, convergence is attained.
Note that even if $c$ were larger, the diffusive coupling alone would not be able to bring the synchronization error to zero (simulations omitted here for the sake of brevity).
Also, the analytical thresholds $c^*$, $c_\R{d}^*$ are conservative.

\section{Conclusions}
This paper solves the problem of achieving asymptotic convergence in networks of heterogeneous nonlinear systems. 
In particular, a distributed approach is proposed that combines traditional diffusive coupling with a discontinuous coupling layer that, under suitable assumptions on the individual dynamics, is capable of guaranteeing asymptotic convergence of all the nodes towards a common trajectory. 
To support the control design, we provided analytical estimates of the minimum coupling gains required to achieve complete synchronization, as a function of the node dynamics, and of the topology of the diffusive and discontinuous layers. The effectiveness of the approach was demonstrated via a representative example.
\begin{figure}[t]
    \centering
    \subfloat[]{\includegraphics[max width=0.48\columnwidth]{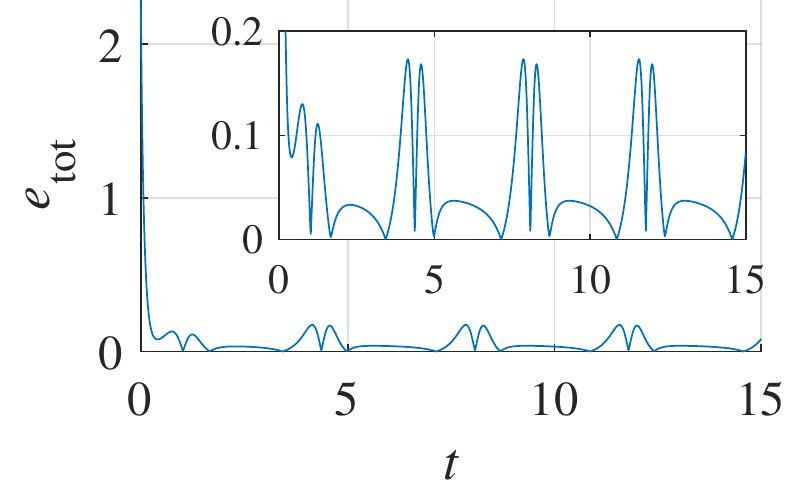}
        \label{fig:oscillators_1}}%
    \subfloat[]{\includegraphics[max width=0.48\columnwidth]{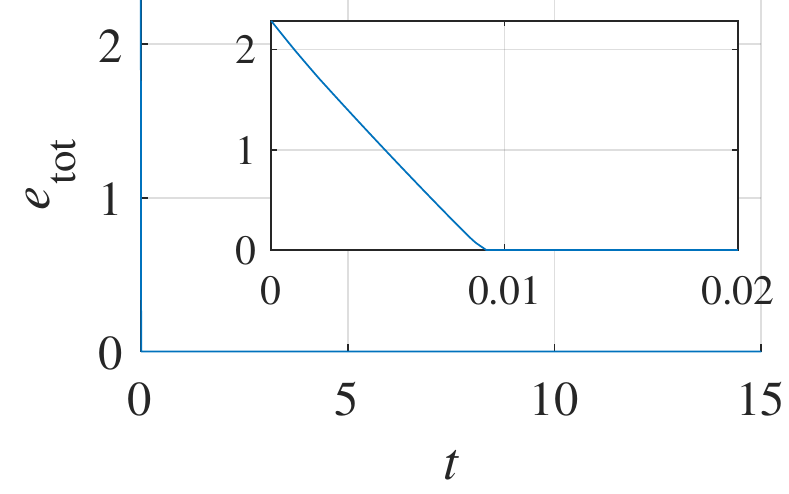}
        \label{fig:oscillators_2}}%
    \caption{Total synchronization error $e_\R{tot} \triangleq \frac{1}{N} \sum_{i = 1}^N \left\lVert \B{e}_i \right\rVert_2$ in a network of three different modified van der Pol oscillators \eqref{eq:modified_van_der_pol}. 
        In (a), $c = 4$, $c_\R{d} = 0$; 
        in (b), $c = 4$, $c_\R{d} = 120$. 
        The initial conditions are $\bar{\B{x}}(t=0)= [ 1.5 \ \ 1.5 \ \ 1.75 \ \ 1.75 \ \ 2 \ \ 2 ]\T$.}
    \label{fig:oscillators}
\end{figure}



\bibliographystyle{IEEEtran}  
\bibliography{references}

\appendix

\begin{lemma}\label{lem:fix_pogromsky_discontinuous}
Consider network \eqref{eq:network}-\eqref{eq:diffusive_discontinuous_coupling}.
If 
\begin{enumerate}[(a)]
    \item all systems in \eqref{eq:network} are strongly strictly semipassive, with stability components $h_i$, $i = 1, \dots, N$;
    \item $c \ge 0$, $c_\R{d} \ge 0$, $\R{sym}( \B{\Gamma}) \ge 0$, and $\mu_\infty^-(\B{\Gamma}_{\R{d}}) \ge 0$;
\end{enumerate} 
then there exists a finite $\bar{\rho} \ge 0$ such that 
\begin{equation}\label{eq:fixed_assumption}
\bar{q}(\bar{\B{x}}) \triangleq \sum_{i = 1}^N \left( h_i(\B{x}_i) - \B{y}_i\T \B{u}_i \right) \ge \bar{\alpha}(\left\lVert \bar{\B{x}} \right\rVert),
\quad \text{if }\left\lVert \bar{\B{x}} \right\rVert \ge \bar{\rho},
\end{equation}
where $\bar{\alpha} : {} [\bar{\rho}, +\infty[ {} \rightarrow \BB{R}_{\ge 0}$ is a continuous and increasing scalar function.
\end{lemma}
\begin{proof} 
    First, it is straightforward to verify that
    \begin{multline}\label{eq:first_part_lemma_disc}
    - \sum_{i = 1}^N \B{y}_i\T \B{u}_i
    = \bar{\B{y}}\T \bar{\B{u}} = \bar{\B{x}}\T \bar{\B{u}} = \\
    c \bar{\B{x}}\T (\B{L} \otimes \B{\Gamma}) \bar{\B{x}}
    + c_\R{d} \bar{\B{z}}\T ( \B{I}_{N_{\C{E}_\R{d}}} \otimes \B{\Gamma}_\R{d} ) \R{sign} (\bar{\B{z}}),
    \end{multline}
    where $N_{\C{E}_\R{d}}$ is the number of edges in $\C{G}_\R{d}$, and $\bar{\B{z}} \triangleq ( \B{B}_{\R{d}}\T \otimes \B{I}_{n} ) \bar{\B{x}}$, with $\B{B}_\R{d}$ being the incidence matrix of $\C{G}_\R{d}$.
    Simple algebraic manipulations show that the first term on the right-hand side of \eqref{eq:first_part_lemma_disc} is non-negative as $c \ge 0$ and $\R{sym}( \B{\Gamma}) \ge 0$.
    By exploiting \cite[Lemma 9]{coraggio2019achieving}, we can also conclude that the second term is non-negative as $c_\R{d} \ge 0$ and $\mu_{\infty}^-(\B{\Gamma}_\R{d}) \ge 0$.  
    To complete the proof, we need to find a scalar $\bar{\rho}$ such that, if $\left\lVert \bar{\B{x}} \right\rVert > \bar{\rho}$,
    it holds that 
    $
    \sum_{i = 1}^N h_i(\B{x}_i) \ge \bar{\alpha}(\left\lVert \bar{\B{x}} \right\rVert).
    $
Such a scalar can be found as follows.
Firstly, note that:
\begin{itemize}[leftmargin=*]
    \item for any $i \in \{1, \dots, N\}$, as $h_i$ is continuous, it is also bounded in the set $\{ \B{x}_i \in \BB{R}^n \mid \left\lVert \B{x}_i \right\rVert \le \rho_i\}$, therefore there exists a finte scalar $H_i \le 0$ such that $h_i(\B{x}_i) \ge H_i$ in that set.
    In addition, $h_i$ is non-negative by definition in $\{ \B{x}_i \in \BB{R}^n \mid \left\lVert \B{x}_i \right\rVert \ge \rho_i\}$; hence,
    \begin{equation}\label{eq:proof_step_06}
    h_i(\B{x}_i) \ge H_i,
    \quad \forall \B{x}_i \in \BB{R}^n;
    \end{equation}
    \item as all systems are strongly strictly semipassive, for each stability component $h_i$ there exists an increasing and radially unbounded function $\alpha_i$ associated to it.
    This implies that, for a given $i \in \{1, \dots, N\}$ and scalar $b$, there exists another scalar $a \ge \rho_i$ such that 
    \begin{equation}\label{eq:proof_step_10}
    \alpha_i(\left\lVert \B{x}_i \right\rVert) > b,
    \quad \text{if }\left\lVert \B{x}_i \right\rVert > a.
    \end{equation}
\end{itemize}
From \eqref{eq:proof_step_10}, there exist $N$ scalars $\rho_i' \ge \rho_i$, for $i = 1, \dots, N$, such that
\begin{equation}\label{eq:proof_step_07}
\alpha_i(\left\lVert \B{x}_i \right\rVert) > - \sum_{\substack{j = 1,j \ne i}}^N H_j,
\quad \text{if } \left\lVert \B{x}_i \right\rVert > \rho_i'.
\end{equation}
Now, define the following partition of $\{1, \dots, N \}$, whose sets are 
$\C{I}_1 \triangleq \{i \mid \left\lVert \B{x}_i \right\rVert \le \rho_i\}$,
$\C{I}_2 \triangleq \{i \mid \rho_i < \left\lVert \B{x}_i \right\rVert \le \rho_i' \}$, and
$\C{I}_3 \triangleq \{i \mid \left\lVert \B{x}_i \right\rVert > \rho_i' \}$.
Then, it is possible to write 
$\sum_{i = 1}^N h_i(\B{x}_i) = 
\sum_{i \in \C{I}_1 \cup \C{I}_2 \cup \C{I}_3} h_i(\B{x}_i)$.
Exploiting \eqref{eq:proof_step_06}, we get
$\sum_{i = 1}^N h_i(\B{x}_i) \ge
\sum_{i \in \C{I}_1} H_i
+ \sum_{i \in \C{I}_2 \cup \C{I}_3} h_i(\B{x}_i);
$
applying \eqref{eq:stability_component}, we have
$ 
\sum_{i = 1}^N h_i(\B{x}_i) \ge
\sum_{i \in \C{I}_1} H_i +
\sum_{i \in \C{I}_2 \cup \C{I}_3} \alpha_i(\left\lVert \B{x}_i \right\rVert)
$. 
Then, we define $\bar{\rho} \triangleq \sqrt{ \sum_{i = 1}^N (\rho_i')^2}$, so that 
\begin{equation}\label{eq:proof_step_08}
\left\lVert \bar{\B{x}} \right\rVert > \bar{\rho} 
\ \ \Rightarrow \ \
\exists i : \left\lVert \B{x}_i \right\rVert > \rho_i' 
\ \ \Leftrightarrow \ \
\C{I}_3 \ne \varnothing.
\end{equation}
For all $\left\lVert \bar{\B{x}} \right\rVert > \bar{\rho}$, we can exploit \eqref{eq:proof_step_07} and \eqref{eq:proof_step_08} to write that
\begin{equation}\label{eq:set_of_inequalities}
\sum_{i = 1}^N h_i(\B{x}_i) \ge
\sum_{i \in \C{I}_1} H_i +
\sum_{i \in \C{I}_2 \cup \C{I}_3} \alpha_i(\left\lVert \B{x}_i \right\rVert)  
 > 0.
\end{equation}

\begin{figure}[t]
    \centering
    \includegraphics[max width=\columnwidth]{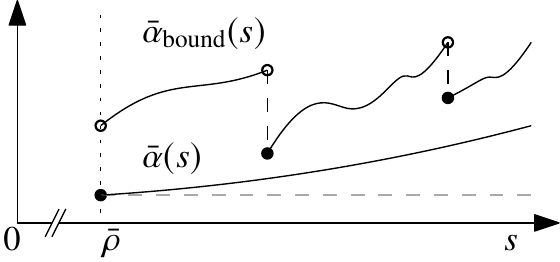}
    \caption{Example of the functions $\bar{\alpha}_{\R{bound}}$ and $\bar{\alpha}$ in the proof of Lemma \ref{lem:fix_pogromsky_discontinuous}.}
    \label{fig:alpha_functions}
\end{figure}
At this point, we define the (not necessarily continuous) positive function $\bar{\alpha}_{\R{bound}} : {} ]\bar{\rho}, +\infty[ {} \rightarrow \BB{R}_{> 0}$ given by
\begin{equation*}
\bar{\alpha}_{\R{bound}}(s) \triangleq \min_{\bar{\B{x}} : \left\lVert \bar{\B{x}} \right\rVert = s} \left(
\sum_{i \in \C{I}_1} H_i +
\sum_{i \in \C{I}_2 \cup \C{I}_3} \alpha_i(\left\lVert \B{x}_i \right\rVert)
\right) > 0.
\end{equation*}
Then, we can define a continuous increasing function $\bar{\alpha} : {} [\bar{\rho}, +\infty[ {} \rightarrow \BB{R}_{\ge 0}$ that satisfies
\begin{equation}
\begin{aligned}
\text{(i)}& \qquad 0 < \bar{\alpha}(s) \le 
\bar{\alpha}_\R{bound}(s), \quad \text{if } s > \bar{\rho}, \\
\text{(ii)}& \phantom{ \qquad 0 < {}} \bar{\alpha}(\bar{\rho}) = \lim_{s \searrow \bar{\rho}} \bar{\alpha}(s);
\end{aligned}
\end{equation}
see Fig. \ref{fig:alpha_functions} for an illustration of $\bar{\alpha}$ and $\bar{\alpha}_\R{bound}$.
From \eqref{eq:set_of_inequalities}, 
\begin{equation}\label{eq:proof_step_09}
\sum_{i = 1}^N h_i(\B{x}_i) \ge \bar{\alpha}(\left\lVert \bar{\B{x}} \right\rVert),
\quad \text{if }\left\lVert \bar{\B{x}} \right\rVert \ge \bar{\rho},
\end{equation}
which, since \eqref{eq:first_part_lemma_disc} is non-negative, proves the Lemma.%
\end{proof} 

\begin{proposition}\label{pro:bounded_jacobian_quad}
    If a function $\B{f} : \BB{R}^n \rightarrow \BB{R}^n$ has an upper bounded Jacobian in $\Omega \subseteq \BB{R}^n$, in the sense that for all $\B{x} \in \Omega$
    \begin{equation}
    {\partial f_i(\B{x})}/{\partial x_i} \le S_{ii},
    \qquad 
    \left\lvert {\partial f_i(\B{x})}/{\partial x_j} \right\rvert \le S_{ij}, \ i \ne j,
    \end{equation} 
    for $S_{ij} \in \BB{R}_{\ge 0}$, $i,j = 1, \dots, n$, then $\B{f}$ is QUAD($\B{I}$, $\B{Q}$) in $\Omega$, with $\B{Q}$ being diagonal and $Q_{ii} = S_{ii} + \sum_{j = 1, j \ne i}^n (S_{ij} + S_{ji})/{2}$.
\end{proposition}
\begin{proof}
    Let us define $\B{x}, \B{\delta} \in \BB{R}^n$, so that $\B{x}, \B{x} + \B{\delta} \in \Omega$.
    From the mean value theorem, there exists $\lambda_i \in [0, 1]$ such that
    $
    f_i(\B{x} + \B{\delta}) - f_i(\B{x}) = \nabla f_i(\B{x} + \lambda_i \B{\delta}) \ \B{\delta}.
    $
    This can be rewritten as
    $f_i(\B{x} + \B{\delta}) - f_i(\B{x}) = 
    \sum_{j = 1}^n \hat{J}_{ij} \delta_j,$
    where $\hat{J}_{ij}={\partial f_i (\B{x} + \lambda_i \B{\delta})}/{\partial x_j}$,
    which, multiplying both sides by $\delta_i$, yields
    \begin{equation}\label{eq:proof_quad_bound_01}
    \delta_i \cdot [f_i(\B{x} + \B{\delta}) - f_i(\B{x})] = 
    \sum_{j = 1}^n \hat J_{ij}  \delta_i \delta_j.
    \end{equation}
    Summing \eqref{eq:proof_quad_bound_01} for $i = 1, \dots, n$, we have
    \begin{equation}
    \B{\delta}\T [\B{f}(\B{x} + \B{\delta}) - \B{f}(\B{x})]
    = \sum_{i = 1}^n \hat{J}_{ii} \delta_i^2 + 
    \sum_{i = 1}^n \sum_{j = 1, j \ne i}^n \hat{J}_{ij} \delta_i \delta_j.
    \end{equation}
    Recalling the expression of the square of a binomial and the bounds on the Jacobian, it holds that
    $
    \hat{J}_{ij} \delta_i \delta_j \le
    \left\lvert \hat{J}_{ij} \delta_i \delta_j \right\rvert \le 
    \left\lvert \hat{J}_{ij} \right\rvert ( \delta_i^2 + \delta_j^2 )/2
    \le S_{ij} ( \delta_i^2 + \delta_j^2 )/2.
    $
    Then, letting $Q_{ii} = S_{ii} + \sum_{j = 1, j \ne i}^n (S_{ij} + S_{ji})/2$, we have
    \begin{multline}
    \B{\delta}\T [\B{f}(\B{x} + \B{\delta} - \B{f}(\B{x})] \le
    \sum_{i = 1}^n \hat{J}_{ii} \delta_i^2 + 
    \sum_{i = 1}^n \sum_{j = 1, j \ne i}^n \frac{S_{ij}}{2} \left( \delta_i^2 + \delta_j^2 \right) \\
    \le \sum_{i = 1}^n S_{ii} \delta_i^2 + 
    \sum_{i = 1}^n \sum_{j = 1, j \ne i}^n \frac{S_{ij}}{2} \left( \delta_i^2 + \delta_j^2 \right) 
    \le
    \sum_{i = 1}^n Q_{ii}
    \delta_i^2.
    \end{multline}
    Defining $\B{y} \triangleq \B{x} + \B{\delta}$, the thesis follows. 
\end{proof}

\end{document}